\documentclass[11pt]{article}
\usepackage[T2A]{fontenc}
\usepackage[cp1251]{inputenc}
\usepackage[english]{babel}
\usepackage{amssymb}
\usepackage{amsmath}
\usepackage{amsthm}
\usepackage{amsfonts}
\usepackage{amssymb}

\newtheorem{theorem}{Theorem}

\newtheorem{lemma}{Lemma}

\newtheorem{proposition}{Proposition}

\newtheorem{definition}{Definition}

\newtheorem{example}{Example} 
\newtheorem{remark}{Remark}

\begin{document}

\begin{center}
\textbf{\uppercase{Levy Laplacian on manifold and Yang-Mills heat flow}}

B.~O.~Volkov

borisvolkov1986@gmail.com

1) Bauman Moscow State Technical University, Moscow, 105005 Russia

2) Steklov Mathematical Institute of Russian Academy of Sciences,\\ Moscow, 119991 Russia
\end{center}

\textbf{Abstract}: A covariant definition of the Levy Laplacian on an infinite dimensional manifold is introduced. It is shown that a time-depended connection in a finite dimensional vector bundle is a solution of the Yang-Mills heat equations
if and only if the associated flow of the parallel transports is a solution of the heat equation for the covariant Levy Laplacian on the infinite dimensional manifold.


key words: Levy Laplacian, Yang-Mills equations, Yang-Mills heat equations, infinite dimensional manifold 

\section*{Intoduction}

In our work we relate  two differential equations of the heat  type: the quasi-linear   Yang-Mills heat equations on a finite-dimensional manifold and the   linear heat equation for the Levy Laplacian on an infinite-dimensional manifold. Namely, we generalize Accardi--Gibilisco--Volovich theorem on the equivalence of the Yang-Mills equations  and the Laplace equation for the Levy Laplacian in the following way: we show that a time-depended connection in a finite-dimensional vector bundle is a solution of the Yang-Mills heat equations if and only if the associated flow of the parallel transports is a solution of the heat equation for the  Levy Laplacian.

The Levy Laplacian is an infinite dimensional Laplacian which  has not any  finite dimensional analogs.
It was introduced by Paul Levy on the functions on $L_2(0,1)$  in the 1920s as follows.
 Let the second derivative of  $f\in C^2(L_2(0,1),\mathbb{R})$ have the form  
 \begin{equation}
 \label{1}
<f''(x)u,v>=\int_0^1\int_0^1 K_V(x;t,s) u(t)v(s)dtds+\int_0^1 K_L(x;t)u(t)v(t)dt,
 \end{equation}
where $K_V(x;\cdot,\cdot)\in L_2([0,1]\times[0,1])$ and $K_L(x;\cdot)\in L_\infty[0,1]$.\footnote{The kernels $K_V(x;\cdot,\cdot)$ and $K_L(x;\cdot)$  are called the Volterra kernel and the L\'{e}vy kernel respectively.}  
Then the Levy Laplacian $\Delta_L$ acts on $f$ by the formula
\begin{equation}
\label{Levy1}
\Delta_Lf(x)=\int_0^1K_L(x;t)dt.
\end{equation}
Another original definition of the Levy Laplacian by Paul Levy is the following.
Let $\{e_n\}$ be an orthonormal basis in $L_2(0,1)$.  Then the Levy Laplacian (generalized by the orthonormal basis $\{e_n\}$) 
acts on $f\in C^2(L_2(0,1),\mathbb{R})$ by the formula
\begin{equation}
\label{Levy2}
\Delta^{\{e_n\}}_L f(x)=\lim_{n\to \infty} \frac 1n \sum_{k=1}^n<f''(x)e_k,e_k>.
\end{equation}
For some orthonormal bases (for example, for $e_n(t)=\sqrt{2}\sin{\pi n t}$) the definitions 
 coincide on the domain of $\Delta_L$ (see~\cite{L1951} and also~\cite{F1986,F2005,KOS}).

The modern situation is the following. The term "Levy Laplacian" is used for various analogs and generalizations of the original Levy Laplacians $\Delta_L$ and $\Delta_L^{\{e_n\}}$. These Levy Laplacians  act on functions (or generalized functions) over different infinite-dimensional spaces.  One of these Levy Laplacians was  introduced by Accardi, Gibilisco and Volovich in~\cite{AGV1993,AGV1994}.   We will denote it  by the symbol $\Delta_L^{AGV}$. The operator $\Delta_L^{AGV}$ was defined by analogy with~(\ref{Levy1}). In~\cite{AGV1993,AGV1994} it was shown that a connection in a  vector bundle over $\mathbb{R}^d$ is a solution of the Yang-Mills equations  if and only if the parallel transport associated with the connection is a solution of the Laplace equation for the Laplacian  $\Delta_L^{AGV}$. The definition of the Levy Laplacian $\Delta_L^{AGV}$ and the theorem on the relationship between the Levy Laplacian and the Yang-Mills equations   was generalized for the case of manifold 
by   Leandre and Volovich in~\cite{LV2001}.   Another definition of the Levy Laplacian on the infinite dimensional manifold was 
introduced by Accardi and Smolyanov in~\cite{AS2006}. In their work the Levy Laplacian was defined as the Cesaro mean of the second order directional derivatives by analogy with~(\ref{Levy2}). The relationship of this Levy Laplacian  and the Yang-Mills equations was studied in~\cite{Volkovdiss}.
The relationship between the Yang-Mills equations and different Levy Laplacians was also studied in~\cite{VolkovLLI,Volkov2017,Volkov2018,VolkovVINITI,Volkov2019}.

In the current paper we introduce the definition of the Levy Laplacian on a manifold in terms of covariant derivatives.  We define this operator as the composition of some infinite dimensional divergence and some nonstandard  gradient. This covariant Levy Laplacian is analog of operator~(\ref{Levy1}). In the flat case its definition coincides with the definition of $\Delta^{AGV}_L$. But in general case its  definition is slightly different from the definition  by Leandre and Volovich from~\cite{LV2001}, which was not in terms of covariant derivatives and was based on the triviality  of the tangent bundle of the base infinite dimensional manifold. However, it seems that the covariant Levy Laplacian, the Levy Laplacian introduced by Leandre and Volovich and the Levy Laplacian introduced by Accardi and Smolyanov coincide on the domain of the first of them.

There are many papers devoted to the heat equations for the Levy Laplacians.
In~\cite{F1986,F2005} some methods of infinite-dimensional analysis were used to study various differential equations with the Levy Laplacian including the heat equation.
 In works~\cite{AS1993,ARS,AB,AS2002} the Levy heat semigroup on the space generalized by the Fourier transforms of the measures on some infinite dimensional space was studied.
The approach to the heat equation for the Levy Laplacian based on the white noise analysis was used in works~\cite{Obata2001,KOS2002,AJS2013} (see also review~\cite{K2003}). In the paper~\cite{AS2006} representations  in  the form of Feynman formulas  for solutions of the heat equation for the Levy Laplacian on  a manifold were obtained. Unfortunately, it  seems that the Levy Laplacian $\Delta_L^{AGV}$, that is connected to the Yang-Mills equations, doesn't coincide with the Laplacians that were used in the mentioned  works, except~\cite{AS2006} (see the discussion in~\cite{Volkov2018}). Is it possible to transfer  the technique of these works for the study of the  Yang-Mills heat equations is an open question.   Some possible ways  for the application of the heat equation for the Levy Laplacian to study  the Yang-Mills equations  are discussed in~\cite{Accardi}.

The Yang-Mills heat flow is a gradient flow for the Yang-Mills action functional. It was introduced by Attiah and Bott in~\cite{AtBo}
and was studied by Donaldson in~\cite{Donaldson} (see also~\cite{DK}). 
If the base manifold is 2-dimensional or 3-dimensional than it is possible to construct a solution of the Yang-Mills equations by solving the Yang-Mills heat equations and letting time tend to infinity (see~\cite{Rade}). In the case of the structure group $U(1)$  the Yang-Mills heat equations are simply the heat equations for 1-forms and a solution of these equations tends to  a harmonic 1-form as time tends to infinity (see~\cite{MR}).  In the dimension four the blow-up does not occur for spherical symmetric solutions (see~\cite{GrotShat}). In the general case the Yang-Mills heat equations have blow-up. 
The dependence of the heat  equation for the Levy Laplacian on the dimension of the base manifold was never studied.
It is interesting to study the behavior of  solutions of this equation in the case then the Yang-Mills heat equation has a  blow-up. 
In~\cite{ABT} the approach to the Yang-Mills heat equations based on the (stochastic) parallel transport was used. But unlike ours,  this  approach  was not based on the Levy Laplacian.
 The Yang-Mills equations and the Yang-Mills heat equations are also related in the following way. In~\cite{Oh1,Oh2} the proof of well-posedness of the Cauchy problem for the  Yang-Mills
equations on the Minkowski space based on the application of the Yang-Mills heat flow was suggested.

The paper is organized as follows. In the first section we give preliminary information from the finite dimensional geometry about  the Yang-Mills heat equation on a time depended connection in the finite dimensional vector bundle. In the second section we give preliminary information from the infinite dimensional geometry about the base Hilbert manifold of the $H^1$-curves. In the third section we introduce the $H^0$-gradient on the space 
of sections in the vector bundle over the base Hilbert manifold of the curves. We consider the parallel transport as a section in this vector bundle and find the value of the $H^0$-gradient on the parallel transport. In the fourth section we define the Levy Laplacian as the composition of the special infinite dimensional divergence and $H^0$-gradient. We find the value of the Levy Laplacian on the parallel transport. In the fifth section we prove the theorem on the equivalence of 
the Yang-Mills heat equations and the heat equation for the Levy Laplacian.

\section{Yang-Mills heat equations}

Bellow, if $\textbf{E}_0$ is a vector bundle over a finite or  infinite dimensional manifold
 $\textbf{M}_0$, the symbol $C^\infty(\textbf{M}_0,\textbf{E}_0)$ denotes the space of  smooth global sections in this bundle
and the symbol $C^\infty(\textbf{W}_0,\textbf{E}_0)$ denotes the space of smooth local  sections on an open set $\textbf{W}_0\subset \textbf{M}_0$. In the infinite-dimensional case derivatives are  understood in the Frechet sense.

Bellow $M$ is a connected smooth compact  $d$-dimensional Riemannian  manifold or $\mathbb{R}^d$. Let $g$ be the Riemannian metric on $M$.
We will raise and lower indices using the metric $g$ and we will sum over repeated indices.
Let $E=E(\mathbb{C}^N,\pi,M,G)$  be a vector bundle over $M$ with the projection $\pi\colon E\to M$ 
and the structure group of $G\subseteq SU(N)$.
 The fiber over $x\in M$ is $E_x=\pi^{-1}(x)\cong\mathbb{C}^N$.
 Let  the Lie algebra of the  structure group
be $Lie (G)\subseteq su(N)$. Let $P$ be the principle bundle  over $M$  associated with $E$ and $ad (P)=Lie(G)\times_G M$   be the adjoint bundle of $P$ (the fiber of $adP$ is
  isomorphic to $Lie (G)$).
A connection $A$ in the vector bundle $E$ is  a smooth section in $\Lambda^1\otimes adP$.
If $W_a$ is an open subset of $M$ and  $\psi_a\colon \pi^{-1}(W_a)=W_a\times \mathbb{C}^N$ 
is a local trivialization of $E$ then in this local trivialization the connection $A$  is a smooth $Lie (G)$-valued 1-form $A^a(x)=A^a_\mu(x)dx^\mu=\psi_{a}A(x)\psi_{a}^{-1}$ on $W_a$.
Let $\psi_a\colon\pi^{-1}(W_a)\cong W_a\times \mathbb{C}^N$ and $\psi_b\colon\pi^{-1}(W_b)\cong W_b\times \mathbb{C}^N$ be two local trivializations of $E$ and $\psi_{ab}\colon W_a\cap W_b\to G$ be the transition function. It means that $\psi_{a}\circ \psi_{b}^{-1}(x,\xi)=(x,\psi_{a b}(x)\xi)$ for all $(x,\xi)\in (W_a\cap W_b)\times \mathbb{C}^N$.
Then for $x\in W_a\cap W_b$ the following holds
\begin{equation}
\label{connection}
A^b(x)=\psi_{ab}^{-1}(x)A^a(x)\psi_{ab}(x)+\psi_{ab}^{-1}(x)d\psi_{ab}(x).
\end{equation}
The connection define the covariant derivative $\nabla$. If $\phi$ is a smooth section in $ad P$ its
covariant derivative has the form
$$
\nabla_\mu\phi=\partial_\mu\phi+[A_\mu,\phi].
$$
The curvature $F$ of the connection $A$ is a smooth section in $\Lambda^2\otimes adP$.
In the local trivialization the curvature $F$ is the $Lie (G)$-valued 2-form $F^a(x)=\sum_{\mu<\nu} F^a_{\mu\nu}(x)dx^\mu\wedge dx^\nu=\psi_{a}F(x)\psi_{a}^{-1}$,
where  $F^a_{\mu\nu}=\partial_\mu A^a_\nu-\partial_\nu A^a_\mu+[A^a_\mu,A^a_\nu]$.
For $x\in W_a\cap W_b$ the following holds
\begin{equation}
\label{curvature}
F^b(x)=\psi_{ab}^{-1}(x)F^a(x)\psi_{ab}(x).
\end{equation}

The Yang-Mills action functional has the form
\begin{equation}
\label{YMaction1}
S_{YM}(A)=-\frac 12\int_{M}tr(F_{\mu\nu}(x)F^{\mu\nu}(x))Vol(dx),
\end{equation}
where $Vol$ is the volume form on the manifold $M$.
The Euler-Lagrange equations for this action functional are
\begin{equation}
\label{YMequations}
\nabla^\mu F_{\mu\nu}=0.
\end{equation}
Locally,
 \begin{equation*}
 \nabla_\lambda F_{\mu
\nu}=\partial_\lambda F_{\mu \nu}+[A_\lambda,F_{\mu
\nu}]-F_{\mu
\kappa}\Gamma^\kappa_{\lambda\nu}-F_{\kappa\nu}\Gamma^\kappa_{\lambda\mu},
\end{equation*}
where $\Gamma^\kappa_{\lambda\nu}$ the Christoffel symbols of the Levy-Civita connection on $M$.
Equations~(\ref{YMequations}) are called  the Yang-Mills equations.

The Yang-Mills heat equations is a nonlinear parabolic differential equation on a time depended connection $A(\cdot,\cdot)\in C^{1,\infty}([0,T]\times M,\Lambda^1\otimes adP)$ (any partial derivative of $A(s,x)$ with respect to the  $x$ variable is jointly $C^1$ on $[0,T]\times M$)  
of the form
\begin{equation}
\label{ymh}
\partial_s A_\nu(s,x)=\nabla^\mu F_{\mu\nu}(s,x).
\end{equation}

For more information about these equations,
in particular, for the initial value problem, the weak solutions, the blow-ups of solutions and the questions related to the gauge choice see~\cite{Rade,ABT,Oh1,Oh2,Gross1}.

\section{Hilbert manifold of $H^1$-curves}

For any sub-interval $I\subset[0,1]$ the symbols  $H^0(I,\mathbb{R}^d)$  and $H^1(I,\mathbb{R}^d)$ denote the spaces of
$L_2$-functions and $H^1$-functions (absolutely continuous with finite energy) on $I$ with values in $\mathbb{R}^d$ respectively. Let  $$\|\gamma\|_0=(\int_I(\gamma(t),\gamma(t))_{\mathbb{R}^d}dt)^{\frac 12}$$ and $$\|\gamma\|_1=(\int_I(\gamma(t),\gamma(t))_{\mathbb{R}^d}dt+\int_I(\dot{\gamma}(t),\dot{\gamma}(t))_{\mathbb{R}^d}dt)^{\frac 12}$$
be the Hilbert norms on these spaces.

The mapping $\gamma\colon [0,1]\to M$ is called $H^1$-curve, if for any interval  $I\subset [0,1]$ and for any  coordinate chart
 $(\phi_a, W_a)$ of the manifold  $M$ such that  $\gamma(I)\subset W_a$, the following holds $\phi_a\circ \gamma\mid_I\in H^1(I,\mathbb{R}^d)$.
Let the symbol $\Omega$ denote the set of all $H^1$-curves in $M$.  For any $x\in M$ let $\Omega_x=\{\gamma\in \Omega\colon \gamma(0)=x\}$ and $\Omega_{x,x}=\{\gamma\in \Omega\colon \gamma(1)=x\}$.

Fix $\gamma\in \Omega$. 
The mapping  $X(\gamma;\cdot)\colon [0,1]\to TM$ such that   $X(\gamma;t)\in T_{\gamma(t)}M$ for any $t\in[0,1]$ is a vector field along  $\gamma$. We will also use the notation $X(\gamma)$ for $X(\gamma;\cdot)$.
Let the symbol $L^\infty_\gamma(TM)$ denote the Banach space of all $L_{\infty}$-fields along $\gamma$. 
The norm  $\|\cdot\|_\infty$  on this space  is defined by
$$
\|X(\gamma)\|_\infty=\mathrm{ess } \sup_{(t\in[0,1])}(\sqrt{g(X(\gamma;t),X(\gamma;t))}).
$$

The symbol  $H^0_\gamma(TM)$ denotes the Hilbert space of all $H^0$-fields along  $\gamma$. The scalar product 
on this space is defined by the formula
\begin{equation}
\label{metrG0}
G_{0}(X(\gamma),Y(\gamma))=\int_0^1g(X(\gamma;t),Y(\gamma;t))dt.
\end{equation}

If $X(\gamma)$ is an absolutely continuous field along  $\gamma\in \Omega$ 
its covariant derivative $\nabla X(\gamma)$ is the field along $\gamma$, 
defined by 
$$
\nabla X(\gamma;t)= \dot{X}(\gamma;t)+\Gamma(\gamma(t))(X(\gamma;t),\dot{\gamma}(t)),
$$
where $(\Gamma(x)(X,Y))^\mu=\Gamma^\mu_{\lambda \nu}(x)X^\lambda Y^\nu$ in local coordinates.
Let $Q(\gamma;\cdot)$ denote the parallel transport generated by the Levi-Civita connection  along the curve  $\gamma$. 
It is easy to show that
$$
\nabla{X}(\gamma;t)=Q(\gamma;t)\frac d{dt}(Q(\gamma;t)^{-1}X(\gamma;t)).
$$
The symbol 
$H^1_\gamma(TM)$ denotes the Hilbert space of all $H^1$-fields along  $\gamma$. The scalar product 
on this space is defined by the formula
\begin{equation}
\label{metrG1}
G_1(X(\gamma),Y(\gamma))=\int_0^1g(X(\gamma;t),Y(\gamma;t))dt
\\+\int_0^1g(\nabla X(\gamma;t),\nabla Y(\gamma;t))dt.
\end{equation}

The set $\Omega$  of all $H^1$-curves in $M$  can be endowed with the structure of a Hilbert manifold in the following way (see~\cite{Driver,Klingenberg,Klingenberg2}).
Let  $d(\cdot,\cdot)$ denote the distance on $M$  generated by the metric   $g$. Let $$W(\gamma,\varepsilon)=\{\sigma\in\Omega\colon d(\gamma(t),\sigma(t))<\varepsilon \text{  for all $t\in[0,1]$}\}.$$
Let  $\widetilde{W}(\gamma,\varepsilon)=\{X\in T_\gamma H^1([0,1],M):\|X\|_\infty<\varepsilon\}$.
Let  $exp_x$ denote the exponential mapping on the manifold $M$ at the point $x\in M$.
For $\gamma\in \Omega$ let  the mapping 
$$
exp_\gamma\colon  \widetilde{W}(\gamma,\varepsilon)\to \Omega
$$
be defined by the formula
$$
exp_\gamma(X)(t)=exp_{\gamma(t)}(X(t)).
$$
It is known that  $exp_\gamma$ is a bijection  between $W(\gamma,\varepsilon)$ and $\widetilde{W}(\gamma,\varepsilon)$. 
The structure of the Hilbert manifold on $\Omega$ is defined by the atlas  $(exp_\gamma^{-1}, W(\gamma,\varepsilon))$. 
The set $\Omega_x$ is a Hilbert  submanifold of $\Omega$ and the set  $\Omega_{x,x}$ is a Hibert submanifold of $\Omega_x$.

We consider two canonical vector bundles $\mathcal H^0$ and $\mathcal H^1$  over the Hilbert manifold $\Omega$
(see~\cite{Klingenberg,Klingenberg2}).
The fiber of $\mathcal H^0$ over  $\gamma\in \Omega$ is the space   $H^0_\gamma(TM)$ and $G_0(\cdot,\cdot)$ is a Riemannian metric on this bundle.
The fiber of  $\mathcal H^1$ over $\gamma\in \Omega$ is the space  $H^1_\gamma(TM)$ and  $G_1(\cdot,\cdot)$ is a Riemannian metric on this bundle. The vector bundle   $\mathcal H^1$ is the tangent bundle over the manifold   $\Omega$. Let $\mathcal H^1_{0,0}$ denote the subbundle of $\mathcal H^1$ such that
 the  fiber of  $\mathcal H^1_{0,0}$  over  $\gamma\in\Omega$ is the space   $\{X\in H^1_\gamma(TM)\colon X(0)=X(1)=0\}$.

A connection  in a vector bundle over an infinite-dimensional manifold  can be given by Christoffel symbols (see~\cite{Klingenberg,Klingenberg2,Lang}). If $\textbf{M}_0$ is a base Hilbert manifold modeled on a Hilbert space $\textbf{H}_0$ and $\textbf{E}_0$ is  a Hilbert vector bundle over $\textbf{M}_0$  with the  fiber  $\textbf{V}_0$ and the projection $\pi_0\colon \textbf{E}_0\to \textbf{M}_0$. If $\textbf{W}_0$ is 
a coordinate chart on $\textbf{M}_0$ then  $\textbf{E}_0$ has a local trivialization $\pi_0^{-1}(\textbf{W}_0)\cong  \textbf{W}_0\times \textbf{H}_0$ and the tangent bundle over $\textbf{M}_0$ has a local trivialization     $T\textbf{W}_0\cong \textbf{W}_0\times\textbf{H}_0$.
Then the Christoffel symbol $\Gamma^{\sim}$ of the connection in $\textbf{E}_0$ is a smooth function on $\textbf{W}_0$ with values in the space of continuous bilinear functionals from $\textbf{V}_0\times \textbf{H}_0$ to $\textbf{V}_0$. Under the coordinate transformations the Christoffel symbols are transformed in the similar way as in the finite-dimensional case.

The Levi-Civita connection on the $d$-dimensional manifold $M$ generates the canonical connection $\nabla^{\mathcal{H}^0}$ in the infinite-dimensional bundle $\mathcal {H}^0$.
 (We associate the connection and the covariant derivative generated by this connection).
Let $\sigma\in \Omega$. 
The Cristoffel symbols  $\Gamma^{\sim}_\sigma$ of the connection $\nabla^{\mathcal{H}^0}$  in the coordinate chart  $(exp_\sigma^{-1}, W(\sigma,\varepsilon))$ are defined as  follows.
 For any  $t\in [0,1]$ we consider the normal coordinate chart on $M$ at the point 
$\sigma(t)$ and the Cristoffel  symbols $\Gamma_{\sigma(t)}$ of the Levi-Civita connection on $M$ in this coordinate chart. If $\gamma\in W(\sigma,\varepsilon)$,
 $X\in H^0_{\gamma}(TM)$ and $Y\in H^1_{\gamma}(TM)$, then $(\Gamma^{\sim}_\sigma(\gamma)(X,Y))(t)\in T_{\gamma(t)}M$ for almost all $t$  is defined by
\begin{equation}
\label{Crisstoffel}
(\Gamma^{\sim}_\sigma(\gamma)(X,Y))(t)=\Gamma_{\sigma(t)}(\gamma(t))(X(t),Y(t)).
\end{equation} 
in the normal coordinate chart on $M$ at the point 
$\sigma(t)$.
In~\cite{Klingenberg,Klingenberg2} it is proved that Cristoffel symbols~(\ref{Crisstoffel}) correctly
define the connection in the vector bundle  $\mathcal{H}^0$. 
Let $X\in C^\infty(W(\sigma,\varepsilon),\mathcal H^0)$ and $Y\in C^\infty(W(\sigma,\varepsilon),\mathcal H^1)$ then  in the normal coordinate chart on $M$ at the point 
$\sigma(t)$ we have the following expression for the covariant derivative
\begin{equation}
\label{nablaalpha0}
\nabla^{\mathcal H^0}_YX(\gamma;t)=d_YX(\gamma,t)+\Gamma_{\sigma(t)}(\gamma(t))(X(t),Y(t)).
\end{equation}

\begin{example}
\label{ex1}
Let  the section $\Upsilon$ in $\mathcal H^0$ be defined by $\Upsilon(\gamma;t)=\dot{\gamma}(t)$.
It  holds that  $\nabla^{\mathcal H^0}_Y\Upsilon(\gamma;t)=\nabla Y(\gamma;t)$ (see~\cite{Klingenberg,Klingenberg2}).
\end{example}

\begin{remark}
The connection $\nabla^{\mathcal H^0}$ is Riemannian. It means that for any smooth local sections $Y,Z$  in $\mathcal H^0$  and smooth local section $X$ in $\mathcal H^1$ the following holds
\begin{equation}
d_XG_0(Y,Z)=G_0(\nabla^{\mathcal{H}^0}_XY,Z)+G_0(Y,\nabla^{\mathcal{H}^0}_XZ).
\end{equation}
\end{remark}

\section{First derivative  and $H^0$-gradient of parallel transport}

  Let $\mathcal E$ be the vector bundle over   $\Omega$, that its fiber  over $\gamma\in \Omega$ is the space $L(E_{\gamma(0)},E_{\gamma(1)})$ of all linear mappings from $E_{\gamma(0)}$ to $E_{\gamma(1)}$.  The parallel 
  transport $U_{1,0}$ generated by the connection $A$ in $E$ can been considered as a section in  $\mathcal E$. 
 Let $\psi_a\colon \pi^{-1}(W_a)\cong W_a\times \mathbb{C}^N$ be a local trivialization of the vector bundle $E$ and let $A^a$ be a local 1-form of the connection $A$ on the open set $W_a\subset M$. For $\gamma\in \Omega$ such that $\gamma([0,1])\subset  W_a$ we can consider the system of differential equations
\begin{equation}
\label{partransp1}
 \left\{
\begin{aligned}
 \frac d{dt}U^a_{t,s}(\gamma)=-A^a_\mu(\gamma(t))\dot{\gamma}^\mu(t)U^a_{t,s}(\gamma)\\
  \frac d{ds}U^a_{t,s}(\gamma)=U^a_{t,s}(\gamma)A^a_\mu(\gamma(s))\dot{\gamma}^\mu(s)\\
\left.U^a_{t,s}\right|_{t=s}=Id.
\end{aligned}
\right.
\end{equation}
Then $U_{1,0}(\gamma)=\psi^{-1}_a U^a_{1,0}(\gamma)\psi_a$ is the parallel transport along $\gamma$  generated by the connection $A$.
If $\gamma([s,t])\in W_a\cap W_b$ and $A^a$ and $A^b$ are the local 1-forms of the connection $A$ on the open sets $W_a$ and $W_b$ respectively then equality~(\ref{connection})
implies that 
\begin{equation}
\label{partransp}
U^a_{t,s}(\gamma)=\psi_{ab}(\gamma(t))U^b_{t,s}(\gamma)\psi_{ba}(\gamma(s)).
\end{equation}
For arbitrary $\gamma\in \Omega$ let consider the partition $c=t_1\leq t_2\leq\ldots t_n=d$ such that
$\gamma([t_i,t_{i+1}])\subset W_{a_i}$  and the family of local trivializations $\psi_{a_i} \colon  \pi^{-1}(W_{a_i})\cong   W_{a_i}\times \mathbb{C}^N$ of the vector bundle  $E$.
 Let
\begin{equation}
\label{paraltrans}
U^{a_{n},a_{1}}_{d,c}(\gamma)=U_{t_{n},t_{n-1}}^{a_{n-1}}
(\gamma)\psi_{a_n a_{n-1}}(\gamma(t_{n-1}))\ldots
U_{t_{3},t_{2}}^{a_{2}}(\gamma)\psi_{a_{2}a_{1}}(\gamma(t_{2}))U_{t_{2},t_{1}}^{a_{1}}(\gamma).
\end{equation}
then $U_{d,c}(\gamma)=\psi^{-1}_{a_n} U^{a_n,a_1}_{d,c}(\gamma)\psi_{a_1}$ and $U_{1,0}(\gamma)$  is a parallel transport along $\gamma$.
By~(\ref{partransp}), the definition of parallel transport does not depend on the choice of the partition and the choice 
of the family of trivializations.
In~\cite{Driver} it is proved that the mapping $\Omega\ni \gamma\to U_{1,0}(\gamma)$ is a smooth section in the vector bundle $\mathcal E$.  The parallel transport does not depend on the choice of parametrization of the curve $\gamma$ 
and $U_{t,s}(\gamma)$ coincide with the parallel transport along the restriction of $\gamma$ on the interval $[s,t]$.
Also, parallel transfer satisfies the multiplicative property:
\begin{equation}
\label{group}
U_{t,s}(\gamma)U_{s,r}(\gamma)=U_{t,r}(\gamma)\text{ for $r\leq s\leq t$}.
\end{equation}

Let $g_{\mathcal E}$ and $g_{\mathcal{H}^0 \otimes \mathcal{E}}$  denote the natural Riemannian metrics in the bundle $\mathcal E$ and $\mathcal{H}^0\otimes \mathcal{E}$ respectively.
 If $X, Y$ are local sections in $\mathcal E$ and $\Phi, \Psi$ are local sections
 in $\mathcal{H}^0$ then
 $$
g_{\mathcal E}(\Phi, \Psi)=-tr(\Phi\Psi),
$$
$$
g_{\mathcal{E}\otimes \mathcal{H}^0}(X\otimes \Phi,Y\otimes \Psi)=G_0(X,Y)g_{\mathcal{E}}(\Phi,\Psi).
$$

\begin{definition}
\label{H_0grad}
The domain $dom\, grad_{H^0}$ of the $H^0$-gradient  consists of all $\varphi\in C^\infty(\Omega,\mathcal{E})$
 such that there exists $J_\varphi\in C^\infty(\Omega,\mathcal{E}\otimes \mathcal{H}^0)$, that 
the following equality holds
$$
g_{\mathcal{E}\otimes \mathcal{H}^0}(J_\varphi(\gamma), X(\gamma)\otimes \Phi(\gamma))=g_{\mathcal{E}}(d_{X} \varphi(\gamma),\Phi(\gamma))
$$
for any $\gamma\in \Omega$, any local smooth section $X$ in $\mathcal{H}_{0,0}^1$ and any local smooth section $\Phi$ in the bundle $\mathcal{E}$.

 The $H^0$-gradient is a linear mapping $grad_{H^0}\colon dom\, grad_{H^0} \to  C^\infty(\mathcal{E}\otimes \mathcal{H}^0)$
defined by the formula
$$
grad_{H^0}\varphi=J_\varphi.
$$
\end{definition}

\begin{remark}
Any connection $B$ in $E$ generates the connection in  $\mathcal E$ such that (see~\cite{LV2001})
$$
\nabla^B_Y \Psi(\gamma)=d_Y\Psi(\gamma)+B_\mu(\gamma(1))Y^\mu(\gamma;1)\Psi(\gamma)-\Psi(\gamma)B_\mu(\gamma(0))Y^\mu(\gamma;0).
$$
If $Y$ is a section in $\mathcal{H}_{0,0}^1$ then $\nabla^B_Y \Psi(\gamma)=d_Y\Psi(\gamma)$. So the definition of the $H^0$-gradient is covariant.
\end{remark}

\begin{example}
\label{ex2}
Let $f\in C^\infty(M,\mathbb{R})$ and
$L^{f}\colon \Omega\to \mathbb{R}$
is defined by
$$
L_f(\gamma)=\int_0^1f(\gamma(t))dt.
$$
Then,
$$
grad_{H^0}L_f(\gamma;t)=\nabla f(\gamma(t)),
$$
where $\nabla$ is the gradient on the manifold $M$.
\end{example}

The following lemma  is Duhamel's Principle (see~\cite{Driver}).
\begin{lemma}
\label{Duhamel}
 Let $V$ be a finite dimensional inner product space. For any $Z\in L_2([c,d],End(V))$ there exists  a unique $P(Z)\in H^1([c,d], End(V))$ such that $\frac {d}{dt}P(Z;t)=-Z(t)P(Z;t)$ for almost all $t$ and $P(Z;c)=Id$.
The mapping $L_2([c,d], End(V))\ni Z\mapsto P(Z;d)$ is $C^\infty$-smooth and
 $P(Z;t)\in Aut(V)$ for all $t\in[c,d]$.
Furthermore, the first derivative of $P$ has the form
\begin{equation}
\label{Duhamels Principle}
<DP(Z;d),\delta Z>=-P(Z;d)\int_{c}^{d}P^{-1}(Z;t)\delta Z(t)P(Z;t)dt.
\end{equation}
\end{lemma}
\begin{proof}
For clarity, we present the idea of the proof. For the complete proof see~\cite{Driver}.
Consider the function $R(t)=P^{-1}(Z_2;t)P(Z_1;t)$ then
\begin{multline}
\label{dR}
\frac {d}{dt} R(t)=\frac {d}{dt} P^{-1}(Z_2,t)P(Z_1;t)+P^{-1}(Z_2,t)\frac {d}{dt}P(Z_1;t)=\\=P^{-1}(Z_2;t)(Z_2(t)-Z_1(t))P(Z_1;t)
\end{multline}
and
\begin{equation}
\label{Pr}
P(Z_2;d)(R(1)-R(0))=P(Z_2;d)(P^{-1}(Z_2;d)P(Z_1;d)-Id)=P(Z_1;d)-P(Z_2;d).
\end{equation}
Together~(\ref{dR}) and~(\ref{Pr}) imply  the formula
$$
P(Z_2;d)-P(Z_1;d)=-P(Z_2;d)\int_c^dP^{-1}(Z_2;t)(Z_2(t)-Z_1(t))P(Z_1;t)dt.
$$
The statement of the proposition can been deduced from this formula.
\end{proof}

\begin{proposition}
\label{prop1deriv}
The first derivative of the parallel  transport has the form
\begin{multline}
\label{firstderpar}
d_XU_{d,c}(\gamma)=-\int_{c}^{d}U_{d,t}(\gamma)F_{\mu \nu}(\gamma(t))X^\mu(\gamma;t)\dot{\gamma}^\nu(t)U_{t,c}(\gamma)dt-\\-A_\mu(\gamma(d))X^\mu(\gamma;d)U_{d,c}(\gamma)
+U_{d,c}(\gamma)A_\mu(\gamma(c))X^\mu(\gamma;c).
\end{multline}

\end{proposition}
\begin{proof}
Consider the partition $c=t_1<t_2<\ldots<t_{n-1}<t_n=d$ and the family of local trivializations $\psi_{a_i}\colon  \pi^{-1}(W_{a_i}) \cong W_{a_i}\times \mathbb{C}^N$ of the vector bundle  $E$ such that
$\gamma([t_{i},t_{i+1}]) \subset W_{a_i}$.
 Lemma~\ref{Duhamel} implies
$$
d_XU^{a_i}_{t_{i+1},t_{i}}(\gamma)=\int_{t_{i}}^{t_{i+1}}U^{a_i}_{t_{i+1},t}(\gamma)(-\partial_{\mu} A^{a_i}_\nu(\gamma(t))X^\mu(\gamma;t)\dot{\gamma}^\nu(t)-A_\mu(\gamma(t))\dot{X}^\mu(\gamma;t))U^{a_i}_{t,t_{i}}(\gamma)dt.
$$
Integrating by parts we have:
\begin{multline}
d_XU^{a_i}_{t_{i+1},t_{i}}(\gamma)=-\int_{t_{i}}^{t_{i+1}}U^{a_i}_{t_{i+1},t}(\gamma)F^{a_i}_{\mu \nu}(\gamma(t))X^\mu(\gamma;t)\dot{\gamma}^\nu(t)U^{a_i}_{t,t_{i}}(\gamma)dt-\\-A^{a_i}_\mu(\gamma(t_{i+1}))X^\mu(\gamma;t_{i+1})U^{a_i}_{t_{i+1},t_{i}}(\gamma)
+U^{a_i}_{t_{i+1},t_{i}}(\gamma)A^{a_{i}}_\mu(\gamma(t_{i}))X^\mu(\gamma;t_{i}).
\end{multline}
Also we have
$$
d_X \psi_{a_{i+1}a_{i}}(\gamma(t_i))=\partial_\mu \psi_{a_{i+1}a_{i}}(\gamma(t_i))X^\mu(\gamma,t_i).
$$
Then
\begin{multline}
\label{dUU}
d_XU^{a_{i+1},a_{i}}_{t_{i+2},t_{i}}(\gamma)=d_X(U^{a_{i+1}}_{t_{i+2},t_{i+1}}(\gamma)\psi_{a_{i+1}a_{i}}(\gamma(t_{i+1}))U^{a_i}_{t_{i+1},t_{i}}(\gamma))=\\
=-\psi_{a_{i+1}}(\int_{t_{i}}^{t_{i+2}}U_{t_{i+1},t}(\gamma)F_{\mu \nu}(\gamma(t))X^\mu(\gamma;t)\dot{\gamma}^\nu(t)U_{t,t_{i}}(\gamma)dt)\psi_{a_{i}}^{-1}-\\
-A^{a_{i+1}}_\mu(\gamma(t_{i+2}))X^\mu(\gamma;t_{i+2})U^{a_{i+1},a_{i}}_{t_{i+2},t_{i}}(\gamma)
+U^{a_{i+1},a_{i}}_{t_{i+2},t_{i}}(\gamma)A^{a_{i}}_\mu(\gamma(t_{i}))X^\mu(\gamma;t_{i})+\\
+U^{a_{i+1}}_{t_{i+2},t_{i+1}}(\gamma)(A^{a_{i+1}}_\mu(\gamma(t_i))\psi_{a_{i+1}a_i}(\gamma(t_i))-\psi_{a_{i+1}a_i}(\gamma(t_i))A^{a_i}_\mu(\gamma(t_i))+\\+\partial_\mu \psi_{a_{i+1}a_{i}}(\gamma(t_i)))X^\mu(\gamma,t_i)U^{a_i}_{t_{i+1},t_{i}}(\gamma).
\end{multline}
Due to~(\ref{connection}) the last summand in~(\ref{dUU}) is equal to zero.
So Leibniz's rule for~(\ref{paraltrans})
 implies 
\begin{multline*}
d_XU^{a_{n},a_{1}}_{d,c}(\gamma)
=-\psi_{a_{n}}(\int_{c}^{d}U_{d,t}(\gamma)F_{\mu \nu}(\gamma(t))X^\mu(\gamma;t)\dot{\gamma}^\nu(t)U_{t,c}(\gamma)dt)\psi_{a_{1}}^{-1}-\\
-A^{a_{n}}_\mu(\gamma(d))X^\mu(\gamma;d)U^{a_{n}a_1}_{d,c}(\gamma)
+U^{a_{n}a_1}_{d,c}(\gamma)A^{a_{1}}_\mu(\gamma(c))X^\mu(\gamma;c)
\end{multline*}
and the statement of the proposition.
\end{proof}

The following proposition is a  direct corollary of Proposition~\ref{prop1deriv}.

\begin{proposition}
The following holds
$$
grad_{H^0}U_{1,0}(\gamma;t)^\mu=-U_{1,t}(\gamma)F^{\mu}_{\ \nu}(\gamma(t))\dot{\gamma}^\nu(t)U_{t,0}(\gamma).
$$
\end{proposition}

\begin{remark}
The first derivative of the parallel transport is well known in literature (see for example~\cite{Gross,Driver}).
The non-commutative Stokes formula is based on formula~(\ref{firstderpar})  for $X(0)=X(1)$  (see~\cite{Arefieva1980} and  also Remark 2.10 in~\cite{Gross}).
\end{remark}

\section{Covariant Levy divergence and Levy Laplacian}

Let $\hat{\otimes}^2T^*M$  and $\wedge^2T^*M$ be the bundles of the symmetic and antisymmetic tensors  of type $(0,2)$ over $M$ respectively. Let  $\mathcal R_1$ be  the vector bundle over  $\Omega$, which fiber over  $
\gamma\in \Omega$ is the space of all $H^0$-sections in $\hat{\otimes}^2T^*M$ along $\gamma$.
Let  $\mathcal R_2$ be  the vector bundle over  $\Omega$, which fiber over  $
\gamma\in \Omega$ is the space of all $H^1$-sections in $\wedge^2T^*M$ along $\gamma$.

Let $C^{\infty}_{AGV}(\Omega,\mathcal H^1_{0,0}\otimes \mathcal H^1_{0,0}\otimes \mathcal E)$ denote the space of all sections $K$ in 
$\mathcal H^1_{0,0}\otimes \mathcal H^1_{0,0}\otimes \mathcal E$ 
that have the form
\begin{multline}
\label{AGVtensors}
K(\gamma)<X,Y>=\int_0^1\int_0^1 K^V(\gamma;s,t)<X(\gamma;t),Y(\gamma;s)>dsdt+\\
+\int_0^1 K^L(\gamma;t)<X(\gamma;t),Y(\gamma;t)>dt+\\
+\frac 12\int_0^1 K^S(\gamma;t)<\nabla X(\gamma;t), Y(\gamma;t)>dt+\\
+\frac 12\int_0^1 K^S(\gamma;t)<\nabla Y(\gamma;t), X(\gamma;t)>dt,
\end{multline}
where $K^L\in C^\infty(\Omega,\mathcal R_1\otimes \mathcal E)$,
$K^S\in C^\infty(\Omega,\mathcal R_2\otimes \mathcal E)$ and
$K^V\in C^\infty(\Omega,\mathcal H^0\otimes \mathcal H^0 \otimes \mathcal{E})$.

\begin{remark}
Tensors of the type~(\ref{AGVtensors}) were in fact  considered by Accardi, Gibilisco and Volovich  in~\cite{AGV1993,AGV1994}.
The kernel $K^V$ is called the Volterra kernel, $K^L$ is called the L\'{e}vy kernel and $K^S$ is called the singular kernel.
By analogy with~\cite{AGV1994,LV2001} it can be proved that these kernels are uniquely defined.
\end{remark}

\begin{definition}
The domain $dom\, div_L$ of the (covariant) Levy divergence consists of all $\psi\in C^\infty(\mathcal{H}^0 \otimes \mathcal{E})$ 
such that there exists $K_\psi\in C^{\infty}_{AGV}(\Omega,\mathcal H^1_{0,0}\otimes \mathcal H^1_{0,0}\otimes \mathcal E)$ that 
the following holds
\begin{equation}
g_{\mathcal H^0\otimes \mathcal{E}} (\nabla^{\mathcal H^0}_{X}\psi(\gamma), Y(\gamma)\otimes \Phi(\gamma))=
g_{\mathcal E}(K_\psi(\gamma)<X(\gamma),Y(\gamma)>,\Phi(\gamma))
\end{equation}
for any $\gamma\in \Omega$, for any   local sections  $X,Y$ in $\mathcal{H}^1_{0,0}$ and any local section $\Phi$ in $\mathcal E$

 The Levy divergence is a linear mapping $div_L\colon dom\, div_L \to  C^\infty(\mathcal E)$
defined by the formula
$$
div_L \psi(\gamma)=\int_0^1K^L_{\psi\,\mu\nu}(\gamma;t)g^{\mu\nu}(\gamma(t))dt,
$$
where $K^L_{\psi}$ is the Levy kernel of the $K_\psi$.
\end{definition}

\begin{remark}
The notion of the Levy divergence was in fact introduced in~\cite{AV1981}. 
See~\cite{Volkov2017,VolkovVINITI,Volkov2019}  for more information about the connection of this divergence  with the Yang-Mills fields.
\end{remark}

\begin{definition}
\label{LevyLapl}
  The value of the Levy Laplacian $\Delta_L$ on $\varphi\in C^\infty(\Omega,\mathcal{E})$
is defined by
\begin{equation}
\label{covLaplaceL}
\Delta_L \varphi=div_L (grad_{H^0} \varphi).
\end{equation}
\end{definition}

\begin{example}
\label{ex3}
Let $L_f\colon \Omega\to \mathbb{R}$ be defined as in Example~\ref{ex2}.
Then
$$
\Delta_L L_f(\gamma)=\int_0^1\Delta_{(M,g)}f(\gamma(t))dt,
$$
where $\Delta_{(M,g)}$ is the Laplace-Beltrami  operator on the manifold $M$.
\end{example}

\begin{theorem}
The following holds
\begin{equation}
\label{LofU}
\Delta_L U_{1,0}(\gamma)=-\int_0^1U_{1,t}(\gamma)\nabla^\mu F_{\mu \nu}(\gamma(t))\dot{\gamma}^\nu(t)U_{t,0}(\gamma)dt.
\end{equation}
\end{theorem}
\begin{proof}
Bellow we denote $grad_{H^0}U_{1,0}(\gamma;t)$ by $J(\gamma;t)$.
At first we find the covariant derivative of the $J$.
In the local coordinates we have the following expression for the directional derivative  $d_YJ$:
\begin{multline*}
d_YJ^\mu(\gamma;t)=-d_YU_{1,t}(\gamma)F^{\mu}_{\ \nu}(\gamma(t))\dot{\gamma}^\nu(t)U_{t,0}(\gamma)
-\\
-U_{1,t}(\gamma)F^{\mu}_{\ \nu}(\gamma(t))\dot{\gamma}^\nu(t)d_YU_{t,0}(\gamma)-\\
-U_{1,t}(\gamma)\partial_\lambda F^\mu_{\ \nu}(\gamma(t))Y^\lambda(\gamma;t)\dot{\gamma}^\nu(t)U_{t,0}(\gamma)-\\
-U_{1,t}(\gamma)F^\mu_{\ \nu}(\gamma(t))\dot{Y}^\nu(\gamma;t)U_{t,0}(\gamma).
\end{multline*}

Let $Y(\gamma;0)=Y(\gamma;1)=0$. Using formulas~(\ref{Crisstoffel}),~(\ref{firstderpar}), 
we obtain that
\begin{multline}
\label{thm1f1}
\nabla^{\mathcal H^0}_YJ(\gamma;t)^\mu=\\
=-U_{1,t}(\gamma)F^\mu_{\ \nu}(\gamma(t))\dot{\gamma}^\nu(t) (\int_0^t U_{t,s}(\gamma)F_{\lambda \kappa}(\gamma(s))Y^\lambda(\gamma;s)\dot{\gamma}^\kappa(s)U_{s,0}(\gamma)ds)-\\-
(\int_t^1 U_{1,s}(\gamma)F_{\lambda \kappa}(\gamma(s))Y^\lambda(\gamma;s)\dot{\gamma}^\kappa(s)U_{s,t}(\gamma)ds)F^\mu_{\ \nu}(\gamma(t))\dot{\gamma}^\nu(t)U_{t,0}(\gamma) -\\-U_{1,t}(\gamma)\nabla_\lambda F^\mu_{\ \nu}(\gamma(t))Y^\lambda(\gamma;t)\dot{\gamma}^\nu(t)U_{t,0}(\gamma)-\\
-U_{1,t}(\gamma) F^\mu_{\ \nu}(\gamma(t))\nabla Y^\nu(\gamma;t)U_{t,0}(\gamma).
\end{multline}
If also $X(\gamma;0)=X(\gamma;1)=0$, the equality
\begin{multline}
\label{thm1f2}
\int_0^1U_{1,t}(\gamma)\nabla_\mu F_{\nu\lambda}(\gamma(t))X^\mu(\gamma;t)Y^\nu(\gamma;t)\dot{\gamma}^\lambda(t)U_{t,0}(\gamma)dt+\\
+\int_0^1U_{1,t}(\gamma) F_{\mu\nu}(\gamma(t))X^\mu(\gamma;t)\nabla Y^\nu(\gamma;t)U_{t,0}(\gamma)=\\
=\frac 12\int_0^1U_{1,t}(\gamma)(\nabla_\mu F_{\nu\lambda}(\gamma(t))\dot{\gamma}^\lambda(t)+\nabla_\nu F_{\mu\lambda}(\gamma(t))\dot{\gamma}^\lambda(t))X^\mu(\gamma;t)Y^\nu(\gamma;t)U_{t,0}(\gamma)dt+\\+\frac 12
\int_0^1U_{1,t}(\gamma)F_{\mu\nu}(\gamma(t))(X^\mu(\gamma;t)\nabla Y^\nu(\gamma;t)+Y^\mu(\gamma;t)\nabla X^\nu(\gamma;t) )U_{t,0}(\gamma)dt
\end{multline}
 can been obtained by integrating by parts, using the Bianchi identities
$$
\nabla_\mu F_{\nu\lambda}+\nabla_\nu F_{\lambda\mu}+\nabla_\lambda F_{\mu\nu}=0
$$
 and renaming of indices.
Formulas~(\ref{thm1f1}) and~(\ref{thm1f2}) together imply that  $J$ belongs to the domain of the Levy divergence. The Volterra kernel $K^V_J$ of $K_J$  has the form
\begin{equation}
\label{KV}
K^V_{J \mu\nu}(\gamma;t,s)=\begin{cases}
U_{1,t}(\gamma)F_{\mu\lambda}(\gamma(t))\dot{\gamma}^\lambda(t)U_{t,s}(\gamma)F_{\nu\kappa}(\gamma(s))
\dot{\gamma}^\kappa(s)U_{s,0}(\gamma),&\text{if $t\geq s$}\\
U_{1,s}(\gamma)F_{\nu\kappa}(\gamma(s))\dot{\gamma}^\kappa(s)U_{s,t}(\gamma)F_{\mu\lambda}(\gamma(t))\dot{\gamma}^\lambda(t)U_{t,0}(\gamma)
,&\text{if $t<s$},
\end{cases}
\end{equation}
the Levy kernel $K^L_J$ has  the form
$$
K^L_{J \mu\nu}(\gamma;t)=\frac 12U_{1,t}(\gamma)(-\nabla_\mu F_{\nu\lambda}(\gamma(t))\dot{\gamma}^\lambda(t)-\nabla_\nu F_{\mu\lambda}(\gamma(t))\dot{\gamma}^\lambda(t))U_{t,0}(\gamma),
$$
ans the singular kernel has the form
$$
K^S_{J \mu\nu}(\gamma;t)=U_{1,t}(\gamma)F_{\mu\nu}(\gamma(t))U_{t,0}(\gamma).
$$
It means that
$$
\Delta_LU_{1,0}(\gamma)=div_LJ(\gamma)=-\int_0^1U_{1,t}(\gamma)\nabla^\mu F_{\mu \nu}(\gamma(t))\dot{\gamma}^\nu(t)U_{t,0}(\gamma)dt.
$$
\end{proof}

\begin{remark}
As it was mentioned in the introduction, the first Levy Laplacian on the infinite dimensional manifold  was introduced  in~\cite{LV2001}.
This Laplacian acts on a space of sections  in a vector bundle   over $\Omega_x$.  The definition of these operator is based on the triviality of the tangent bundle over $\Omega_x$.
In the case $M=\mathbb{R}^d$ both this Levy Laplacian and the covariant Levy Laplacain~(\ref{covLaplaceL}) coincide
with the Levy Laplacian introduced in~\cite{AGV1993}.
The~Levy Laplacian as the Cesaro mean of the second order directional derivatives was defined on a space of sections  in a vector bundle   over $\Omega_x$  in~\cite{AS2006}.
The values of Levy Laplacians introduced in~\cite{LV2001,AS2006}  on the parallel transport coincide with~(\ref{LofU}) (see~\cite{LV2001,Volkovdiss}). Definitions~\ref{H_0grad} and~\ref{LevyLapl} of the  $H^0$-gradient and the Levy Laplacian  can be transferred to the infinite dimension bundles over $\Omega_x$.
In this case we conjecture that all three Levy Laplacians coincide on the domain of the covariant Levy Laplacian~(\ref{covLaplaceL}) .
\end{remark}

\begin{remark}
Laplacians on  abstract Hilbert manifolds were considered in the  literature (see~\cite{BP2016}). It is interesting is it possible to define the Levy Laplacian on the abstract Hilbert manifold and to study the heat equation for this operator. It seems that  the definition of the Levy Laplacian as the Cesaro mean of the second order directional derivatives (see~\cite{AS2006}) can be useful for this purpose.

In the work~\cite{AS2006} the Feynman approximation was obtained for  the solution of the heat equation for the Levy Laplacian.
It is interesting whether it is possible to develop a related approach of the quasi-Feynman approximations to this equation (see~\cite{Remizov}).

Due to the fact that the Levy Laplacian can be defined as the averaging of finite-dimensional Laplacians, it would be interesting to investigate whether it is possible to obtain the heat semigroup for the Levy Laplacian by averaging  of the semigroups for these finite-dimensional
operators (for the method of the averaging  of semigroups see~\cite{OSS2016,Sakbaev2017}).  
\end{remark}

\section{Heat equation for Levy Laplacian and Yang-Mills heat equations}

In  this section  $A(\cdot,\cdot)\in C^{1,\infty}([0,T]\times M,\Lambda^1\otimes ad P)$
and $U_{1,0}(s,\gamma)$ is the parallel transport generated by the connection $A(s,\cdot)$
along the curve $\gamma\in \Omega$.
\begin{proposition}
For any $\gamma\in \Omega$ the following holds
\begin{equation}
\partial_s U_{1,0}(s,\gamma)=-\int_0^1U_{1,t}(s,\gamma)\partial_sA_\mu(s,\gamma(t))\dot{\gamma}^\mu(t)U_{t,0}(s,\gamma)dt.
\end{equation}
\end{proposition}
\begin{proof}
Consider the partition $0=t_1<t_2<\ldots<t_{n-1}<t_n=1$ and the family of local trivializations $\psi_{a_i}\colon  \pi^{-1}(W_{a_i}) \cong W_{a_i}\times \mathbb{C}^N$ of the vector bundle  $E$ such that
$\gamma([t_{i},t_{i+1}]) \subset W_{a_i}$. Due to the fact that the time-depended connection belongs to the class $C^{1,\infty}$,  the mapping 
 $$[0,T]\ni s\mapsto A^{a_{i}}_\mu(s,\gamma(\cdot))\dot{\gamma}^\mu(\cdot)\in L_2([t_{i},t_{i+1}],Lie(G))$$ 
 is  differentiable for any $i\in\{1,\ldots,n\}$.   Lemma~\ref{Duhamel} implies
$$
\partial_sU^{a_{i}}_{t_{i+1},t_{i}}(s,\gamma)=\int_{t_{i}}^{t_{i+1}}U^{a_{i}}_{t_{i+1},t}(s,\gamma)(-\partial_s A^{a_{i}}_\nu(s,\gamma(t))\dot{\gamma}^\nu(t))U^{a_{i}}_{t,t_{i}}(s,\gamma)dt.
$$
Then Leibniz's rule for~(\ref{paraltrans}) implies the statement of the proposition.
\end{proof}

\begin{theorem}
\label{maintheorem}
The following two assertions are equivalent:

1) the flow of connections $[0,T]\ni s\mapsto A(s,\cdot)$ is a solution of the Yang-Mills  heat equations~(\ref{ymh}):

2) the flow of parallel transports   $[0,T]\ni s\mapsto U_{1,0}(s,\cdot)$  is a solution 
of the heat equation for the Levy Laplacian:
\begin{equation}
\label{llhU}
\partial_sU_{1,0}(s,\gamma)=\Delta_L U_{1,0}(s,\gamma).
\end{equation}
\end{theorem}

\begin{proof}
Let the flow of the parallel transports is a solution of the heat equation for the Levy Laplacian. Fix any curve $\gamma\in C^1([0,1],M)$. Let  the curve $\gamma^r\in\Omega$ be defined by $$\gamma^r(t)=\begin{cases}
\gamma(t), &\text{if $t\leq r$,}\\
\gamma(r), &\text{if $t>r$.}
\end{cases}$$

Let us introduce the function $R\in C^1([0,1],L(E_{\gamma(0)},E_{\gamma(1)}))$ by the formula:
\begin{equation}
R(r)=U_{1,r}(s,\gamma)(\partial_sU_{1,0}(s,\gamma_r)-\Delta_L U_{1,0}(s,\gamma_r)).
\end{equation}
Due to the invariance with respect to the reparametrization  of the parallel transport and due to the  multiplicative property~(\ref{group}) we have
\begin{equation}
\label{R}
R(r)=\int_0^r U_{1,t}(s,\gamma)(\partial_s A_\nu(s,\gamma(t))\dot{\gamma}^\nu(t)-\nabla_\mu F^\mu_{\ \nu}(s,\gamma(t))\dot{\gamma}^\nu(t))U_{t,0}(s,\gamma)dt.
\end{equation}

If $U_{1,0}(\cdot,\cdot)$ is a solution of~(\ref{llhU}) then  $\partial_sU_{1,0}(s,\gamma_r)-\Delta_L U_{1,0}(s,\gamma_r)=0$ and, therefore,  $R(r)\equiv 0$. Differentiating~(\ref{R}), we obtain
$$
\frac {d}{dr}R(r)= U_{1,r}(s,\gamma)(\partial_s A_\nu(s,\gamma(r))\dot{\gamma}^\nu(r)-\nabla_\mu F^\mu_{\ \nu}(s,\gamma(r))\dot{\gamma}^\nu(r))U_{r,0}(s,\gamma)\equiv 0.
$$
It means that
$$
\partial_s A_\nu(s,\gamma(r))\dot{\gamma}^\nu(r)-\nabla_\mu F^\mu_{\ \nu}(s,\gamma(r))\dot{\gamma}^\nu(r)=0
$$
for all $\gamma\in C^1([0,1],M)$ and for all $r\in [0,1]$. So $A(s,\cdot)$ is the Yang-Mills heat flow. The other side of the theorem is trivial.
\end{proof}

\begin{remark}
If the connection $A$ is  time-independent, Theorem~\ref{maintheorem} becomes the Accardi-Gibilisco-Volovich  theorem
on the equivalence of the Yang-Mills equations  and the Laplace equation for the Levy Laplacian. 
\end{remark}

\begin{remark}
Let $f(s,\cdot)$ be a solution of the heat equation on the manifold $M$:
$$
\partial_s f(s,\cdot)=\Delta_{M,g} f(s,\cdot).
$$
Let the family of functionals $L_{f(s,\cdot)}$ on $\Omega$ be defined as in  Examples~\ref{ex2} and~\ref{ex3}.
Then $L_{f(s,\cdot)}$ is a solution of the heat equation for the Levy Laplacian:
$$
\partial_sL_{f(s,\cdot)}=\Delta_L L_{f(s,\cdot)}.
$$
\end{remark}

\begin{remark}
The definitions of the  $H^0$-gradient and the Levy Laplacian can be transferred to the infinite dimension bundle over $\Omega_{x,x}$.  In this case these definitions have the simplest form. We don't know whether Accardi-Gibilisco-Volovich theorem holds in this case: is it true that if $\Delta^LU(\gamma)=0$ for any  $\gamma\in \Omega_{x,x}$ than the  connection associated  with this parallel transport $U$  is a solution of the Yang-Mills equations.
Our proof of the Theorem~\ref{maintheorem} essentially uses the fact that the endpoints of the curves from the base manifold are not fixed.

In this context the following result is interesting.
In the work~\cite{Driver}  it is shown that if an operator-valued function on $\Omega_{x,x}$ has some properties 
of the parallel transport (smoothness, group property, invariance with respect to reparametrization) it is truly the parallel transport generated by some connection in $E$. For the generalization of this result for a groupoid see~\cite{Gibilisco1997}.

\end{remark}



\section*{Acknowledgments}

The author would like to express his deep gratitude to L.~Accardi, O.~G.~Smolyanov and I.~V.~Volovich for helpful discussions.

\section*{Funding}
This work was supported by the Russian Science Foundation under grant 19-11-00320.


\begin{thebibliography}{2}




\bibitem{Accardi} L.~Accardi,  Yang-Mills Equations and Levy-Laplacians, (Dirichlet Forms and Stochastic Processes, Beijing, 1993. de Gruyter, 1995), pp. 1--24.

\bibitem{AB} 
L.~Accardi and V.~I. Bogachev, \textquotedblleft The Ornstein-Uhlenbeck process associated with the Levy-Laplacian and  its  Dirihlet form,\textquotedblright   Probab. Math. Statist.   \textbf{17}(1),  95--114 (1997).

\bibitem{AGV1993}
L.~Accardi, P.~Gibilisco and  I.~V. Volovich,  \textquotedblleft The Levy Laplacian and the Yang-Mills equations,\textquotedblright  Rend. Lincei. Sci.
Fis. Nat.  \textbf{4}  (3), 201--206 (1993).


 
\bibitem{AGV1994}
L.~Accardi, P.~Gibilisco and  I.~V. Volovich, \textquotedblleft Yang-Mills gauge fields as harmonic functions for the Levy Laplacian,\textquotedblright  
Russ. J. Math. Phys.  \textbf{2} (2), 235--250 (1994).




\bibitem{AJS2013} 
L.~Accardi, U.~C. Ji and K. Saito,  \textquotedblleft The exotic (higher order Levy) Laplacians generate the Markov processes given
by distribution derivatives of white noise,\textquotedblright   Infin. Dimens. Anal. Quantum Probab. Relat. Top.  \textbf{16} (3), Pap. 1350020 (14 pages) (2013).


\bibitem{ARS}
L.~Accardi, P. Rozelli and O.~G. Smolyanov, \textquotedblleft Brownian motion generated by the Levy Laplacian,\textquotedblright   Math. Notes \textbf{54} (5), 1174--1177 (1993).



\bibitem{AS1993}
L.~Accardi and O.~G. Smolyanov [Smolyanov], \textquotedblleft On Laplacians and traces,\textquotedblright   Conf.
Semin. Univ. Bari.  \textbf{250}, 1--25 (1993).

\bibitem{AS2002}
L.~Accardi and O.~G. Smolyanov, \textquotedblleft Levy-Laplace Operators in Functional Rigged Hilbert Spaces,\textquotedblright   Math. Notes \textbf{72}(1), 129--134 (2002).


\bibitem{AS2006} 
L.~Accardi and O. G. Smolyanov, \textquotedblleft Feynman formulas for evolution equations with Levy Laplacians on infinite-dimensional manifolds,\textquotedblright  Doklady Mathematics   \textbf{73}(2),  252--257 (2006).
  
  
\bibitem{AtBo} 
M. Atiyah and R.  Bott, \textquotedblleft The Yang Mills equations over Riemannian surfaces.\textquotedblright Philos. Trans, R. Soc.
London A \textbf{308}, 524--615 (1982)

  
\bibitem{Arefieva1980}
 I.~Ya. Aref'eva, \textquotedblleft Non-Abelian Stokes formula,\textquotedblright   Theoret. and Math. Phys. \textbf{43}(1), 353--356 (1980).
  
  
\bibitem{AV1981}
I.~Ya.~Aref'eva and I.\,V.~Volovich,\textquotedblleft Higher order functional conservation laws in gauge
theories,\textquotedblright   Proc. Int. Conf. Generalized Functions and their Applications in Mathematical Physics
(Academy of Sciences of the  USSR, 1981,  43--49. [In Russian].  




\bibitem{ABT}
 M.~Arnaudon, R.~O. Bauer and A.~Thalmaier,\textquotedblleft  A probabilistic approach to the Yang-Mills heat equation,\textquotedblright J. Math. Pures Appl. {\small \bf 81},   143--166 (2002). 






  


\bibitem{BP2016}
 Yu.~V. Bogdanskii and A.~Yu. Potapenko, \textquotedblleft Laplacian with respect to a measure on the Riemannian manifold and the Dirichlet problem. I,\textquotedblright  Ukr. Math. J. \textbf{68}(7), 1021--1033 (2016).


 

 
 
\bibitem{Donaldson}
 S.~K. Donaldson, \textquotedblleft Anti self-dual Yang-Mills connections over complex
algebraic surfaces and stable vector bundles,\textquotedblright Proc. London Math. Soc.
(3) \textbf{50} (1), 1--26 (1985). 




\bibitem{DK} 
S.K. Donaldson and P.B. Kronheimer, \emph{The Geometry of Four-Manifolds}, (Clarendon Press, New York, 1990).
 

\bibitem{Driver}
 B.~Driver, \textquotedblleft Classifications of Bundle Connection Pairs
by Parallel Translation and Lassos,\textquotedblright J. Funct. Anal. \textbf{83}, 185--231 (1989).



\bibitem{F1986}
 M.~N. Feller, \textquotedblleft Infinite-dimensional elliptic equations and operators of Levy types,\textquotedblright  Russian Math. Surveys, \textbf{41} (4) 119--170 (1986).

\bibitem{F2005}
 M.~N. Feller, \emph{The Levy Laplacian} (Cambridge Tracts in Math., 166,  Cambridge, Cambridge Univ. Press, 2005).



\bibitem{Gibilisco1997}
P. Gibilisco, \textquotedblleft Bundle-connection pairs and loop group representations, \textquotedblright  Math. Notes \textbf{61}(4), 417--429 (1997).



\bibitem{Gross}
 L.~Gross, \textquotedblleft A Poincar\`{e} lemma for connection forms,\textquotedblright J. Funct. Anal. \textbf{63}, 1--46 (1985).

\bibitem{Gross1} L.~Gross, \textquotedblleft  Stability of the Yang-Mills heat equation for finite action,\textquotedblright arXiv:1711.00114 (2016).

\bibitem{GrotShat}
J. F. Grotowski  and J. Shatah, \textquotedblleft Geometric evolution equations in critical dimensions,\textquotedblright  Calculus of Variations and Partial Differential Equations \textbf{30} (4), 499--512 (2007). 


\bibitem{Lang}
 S.~Lang, \emph{Differential Manifolds} (Springer-Verlag, Berlin-Heidelberg
New York, 1985).

\bibitem{Klingenberg} 
W.~Klingenberg, \emph{Riemannian geometry}, (de Gruyter Studies in Mathematics, Vol. 1, Berlin, 1982).

\bibitem{Klingenberg2} 
 W.~Klingenberg, \emph{Lectures on closed geodesics}, (Grundlehren der Mathematischen Wissenschaften, Vol. 230. Springer-Verlag, Berlin-New York, 1978).






\bibitem{KOS}
H.-H. Kuo, N.~Obata and K.~Sait\^{o}, \textquotedblleft  L\'{e}vy-Laplacian of Generalized
Functions on a Nuclear Space,\textquotedblright J. Funct. Anal.   \textbf{94} (1), 74--92 (1990).



\bibitem{KOS2002}
H.-H. Kuo, N.~Obata and K.~Sait\^{o},  \textquotedblleft  Diagonalization of the Levy Laplacian
and related stable processes,\textquotedblright Infin. Dimens. Anal. Quantum Probab. Relat. Top. \textbf{5}(3), 317--331 (2002). 

\bibitem{K2003}
H.-H. Kuo, \textquotedblleft  Recent Progress on the White Noise Approach to the Levy Laplacian,\textquotedblright Conference: Proceedings of the Meijo Winter School 2003 - Quantum Information and Complexity, 267--295 (2004). 




\bibitem{LV2001}
R.~L\'{e}andre and I.~V. Volovich,  \textquotedblleft  The Stochastic Levy Laplacian and Yang--Mills equation on manifolds,\textquotedblright Infin. Dimens. Anal. Quantum Probab. Relat. Top. \textbf{4} (2), 161--172  (2001).



\bibitem{L1951} P.~Levy,  {\it Probl\`{e}mes concrets d'analyse fonctionnelle},  (Paris,
Gautier-Villars, 1951).



\bibitem{MR}
 A. Milgram and P. Rosenbloom, \textquotedblleft  Harmonic forms and heat conduction,\textquotedblright Proc. Nat. Acad. Sei. U.S.A. \textbf{37}, 180--184 (1951).




\bibitem{Oh1}
S.-J. Oh, Ph.D. dissertation, Princeton University, Princeton, N.J. (2013).


\bibitem{Oh2}
 S.-J. Oh, \textquotedblleft  Finite energy global well-posedness of the Yang-Mills equations
on R1+3: an approach using the Yang-Mills heat flow,\textquotedblright Duke Math. J.
\textbf{164}(9), 1669--1732 (2015). 



\bibitem{Obata2001}
N.~Obata, \textquotedblleft  Quadratic Quantum White Noises and L\'{e}vy-Laplacians,\textquotedblright
 Nonlinear Analysis-Theory
Methods and Applications  {\bf 47},   2437--2448 (2001). 





\bibitem{OSS2016}
U.~N. Orlov, V.~Zh. Sakbaev and O~G. Smolyanov, \textquotedblleft Unbounded random operators and Feynman formulae,\textquotedblright Izv. Math \textbf{80} (6), 1131--1158 (2016).




\bibitem{Rade} 
J.~Rade, \textquotedblleft On the Yang-Mills heat equation in two and three dimensions,\textquotedblright J. Reine Angew. Math. \textbf{431},  123--163 (1992).

\bibitem{Remizov}
 I.~D. Remizov, \textquotedblleft Quasi-Feynman formulas --- a method of obtaining the evolution operator for the Schrodinger equation,\textquotedblright J. Funct. Anal. \textbf{270} (12), 4540--4557 (2016).

\bibitem{Sakbaev2017}
 V.~Zh. Sakbaev, \textquotedblleft Averaging of random walks and shift-invariant measures on a Hilbert space,\textquotedblright Theoret. and Math. Phys. \textbf{191} (3), 886--909   (2017).











\bibitem{Volkovdiss}
B.~O. Volkov, Ph.D. diss., Department of Mechanics and Mathematics, Lomonosov Moscow State University, Moscow (2014) [in Russian].



\bibitem{VolkovLLI} 
B.~O. Volkov, \textquotedblleft Levy Laplacians and instantons,\textquotedblright Proc. Steklov Inst. Math. \textbf{290}, 210--222 (2015). 


\bibitem{Volkov2017}
B.~O. Volkov, \textquotedblleft Stochastic Levy Differential  Operators and Yang-Mills Equations,\textquotedblright Infin. Dimens. Anal. Quantum Probab. Relat. Top. \textbf{20} (2), Pap. 1750008 (23 pages)  (2017).

\bibitem{Volkov2018}
B.~O. Volkov, \textquotedblleft Levy Laplacians in Hida calculus and Malliavin calculus,\textquotedblright  Proc. Steklov Inst. Math. \textbf{301}, 11--24 (2018). 

\bibitem{VolkovVINITI}
B.~O. Volkov, Applications of Levy Differential Operators in the Theory of Gauge Fields, Quantum probability, Itogi Nauki i Tekhniki. Ser. Sovrem. Mat. Pril. Temat. Obz., 151, VINITI, Moscow, 2018, 21--36 [in Russian].


\bibitem{Volkov2019}
B.~O. Volkov, \textquotedblleft Levy Differential Operators and Gauge Invariant Equations for Dirac and Higgs Fields,\textquotedblright Infin. Dimens. Anal. Quantum
Probab. Relat. Top. \textbf{22} (1),  Pap. 1950001 (20 pages) (2019).



\end{thebibliography}
\end{document}